\documentclass{article}
\usepackage[colorlinks,citecolor=blue,bookmarks=true]{hyperref}
\usepackage{amsmath} 
\usepackage{amsthm} 
\usepackage{amssymb}	
\usepackage{graphicx} 
\usepackage{multirow}
\usepackage{color}
\usepackage[dvips,letterpaper,margin=1in,bottom=1in]{geometry}
\usepackage[capitalize,noabbrev]{cleveref}
\usepackage{subcaption}

\usepackage[utf8]{inputenc}
\usepackage[english]{babel}

\usepackage{algorithm}
\usepackage{algpseudocode}
\makeatletter
\renewcommand{\ALG@name}{Protocol}
\makeatother

\usepackage{bm}
\usepackage{bbm}
\usepackage{diagbox}
\usepackage{mathtools}
\usepackage{ytableau}

\newtheorem{theorem}{Theorem}[section]

\newtheorem{lemma}[theorem]{Lemma}
\newtheorem{corollary}[theorem]{Corollary}

\newtheorem{fact}[theorem]{Fact}

\newtheorem{definition}[theorem]{Definition}

\newtheorem{problem}[theorem]{Problem}

\crefname{question}{Question}{Questions}

\newcommand{\braket}[2]{\left< #1 \vphantom{#2} \middle| #2 \vphantom{#1} \right>} 

\DeclarePairedDelimiter\ket{\lvert}{\rangle}
\DeclarePairedDelimiter\bra{\langle}{\rvert}

\newcommand{\E}{\mathop{\bf E\/}}

\newcommand{\polylog} {\operatorname{polylog}}

\newcommand{\Sh} {\textup{Sh}}

\newcommand{\ketbra}[2]{\ensuremath{\ket{#1}\!\bra{#2}}}
\renewcommand{\braket}[2]{\ensuremath{\langle {#1} \vert {#2} \rangle}}

\allowdisplaybreaks

\title{Inverse-free quantum state estimation with Heisenberg scaling}
\date{}

\author{Kean Chen\thanks{University of Pennsylvania, Philadelphia, USA. Email: \texttt{keanchen.gan@gmail.com}}}

\begin{document}

\maketitle

\begin{abstract}
In this paper, we present an inverse-free pure quantum state estimation protocol that achieves Heisenberg scaling. Specifically, let $\mathcal{H}\cong \mathbb{C}^d$ be a $d$-dimensional Hilbert space with an orthonormal basis $\{\ket{1},\ldots,\ket{d}\}$ and $U$ be an unknown unitary on $\mathcal{H}$. Our protocol estimates $U\ket{d}$ to within trace distance error $\varepsilon$ using $O(\min\{d^{3/2}/\varepsilon,d/\varepsilon^2\})$ forward queries to $U$. This complements the previous result $O(d\log(d)/\varepsilon)$ by 
\hyperlink{cite.van2023quantum}{\textcolor{blue}{van Apeldoorn, Cornelissen, Gily\'en, and Nannicini (SODA 2023)}}, 
which requires both forward and inverse queries. Moreover, our result implies a query upper bound $O(\min\{d^{3/2}/\varepsilon,1/\varepsilon^2\})$ for inverse-free amplitude estimation, improving the previous best upper bound $O(\min\{d^{2}/\varepsilon,1/\varepsilon^2\})$ based on optimal unitary estimation by \hyperlink{cite.haah2023query}{\textcolor{blue}{Haah, Kothari, O'Donnell, and Tang (FOCS 2023)}}, and disproving a conjecture posed in \hyperlink{cite.tang2025amplitude}{\textcolor{blue}{Tang and Wright (2025)}}. 
\end{abstract}

\section{Introduction}

Given access to an unknown $d$-dimensional quantum state $\ket{\psi}$, how efficiently can we estimate $\ket{\psi}$ to within trace distance error $\varepsilon$?
This state estimation problem has been extensively studied in quantum information and quantum learning theory, and can be categorized into two settings.
In the sample-access setting, one assumes access to multiple copies of the unknown state $\ket{\psi}$ and the sample complexity has been shown to be $\Theta(d/\varepsilon^2)$~\cite{bruss1999optimal,keyl1999optimal,werner1998optimal} (see also \cite{keyl2006quantum,haah2016sample,o2016efficient,o2017efficient,guctua2020fast,yuen2023improved,scharnhorst2025optimal,pelecanos2025debiased} for the mixed state case).
In contrast, the query-access setting assumes access to a state-preparation unitary $U$ of $\ket{\psi}$, as well as to its inverse $U^{-1}$. In this model, the query complexity is $\widetilde{\Theta}(d/\varepsilon)$\footnote{$\widetilde{\Theta}$ hides logarithmic factors.}~\cite{van2023quantum}.
The key distinction between these two models lies in their error scaling: the query model achieves the Heisenberg scaling $\Theta(1/\varepsilon)$, whereas the sample model is limited to the classical scaling $\Theta(1/\varepsilon^2)$.

Recently, Tang and Wright~\cite{tang2025amplitude} established that, in the absence of inverse query access, quantum amplitude estimation requires a query complexity of $\Omega(\min\{d, 1/\varepsilon^2\})$. This stands in sharp contrast to the well-known upper bound of $O(1/\varepsilon)$~\cite{brassard2000quantum}, which assumes access to both forward and inverse queries. It therefore reveals that disallowing inverse access incurs an additional $\Omega(d)$ overhead in order to go beyond the classical scaling $\Theta(1/\varepsilon^2)$. 

Conversely, for the unitary estimation task, Haah, Kothari, O'Donnell, and Tang~\cite{haah2023query} showed that inverse access is not necessary to achieve the optimal query complexity $\Theta(d^2/\varepsilon)$ (see also the work by Yang, Renner, and Chiribella~\cite{yang2020optimal} for unitary estimation under an average-case distance).
Based on this, \cite{tang2025amplitude} conjectured that the query complexity of inverse-free amplitude estimation is $\Theta(\min\{d^2/\varepsilon,1/\varepsilon^2\})$, suggesting that achieving inverse-free Heisenberg scaling requires estimating the entire unitary.
The intuition behind their conjecture is a two-regime characterization of the inverse-free setting:
1) \textit{estimating all features of a unitary} can be done coherently with a $O(1/\varepsilon)$ scaling;
2) \textit{estimating one desired feature of the unitary} cannot be done better than the naive incoherent $O(1/\varepsilon^2)$ algorithm.

\subsection{Our results}
In this paper, we study the quantum state estimation problem in the query access model.
We present an inverse-free state estimation protocol with query complexity $O(\min\{d^{3/2}/\varepsilon,d/\varepsilon^2\})$. 
This further implies a query upper bound $O(\min\{d^{3/2}/\varepsilon,1/\varepsilon^2\})$ for inverse-free amplitude estimation by combining it with the naive $O(1/\varepsilon^2)$ amplitude estimation. 
Our results demonstrate that, for both state and amplitude estimation, Heisenberg scaling can be achieved without incurring a $\Theta(d^2)$ overhead in the inverse-free setting, and thereby disproving the conjecture proposed in \cite{tang2025amplitude} (see \cref{sec-10271551} for further discussions).
We formalize the problem and our main results as follows.

\begin{problem}[State estimation]
Let $\mathcal{H}\cong \mathbb{C}^d$ be a $d$-dimensional Hilbert space with an orthonormal basis $\{\ket{1},\ket{2},\ldots,\ket{d}\}$.
Suppose $U$ is an unknown unitary on $\mathcal{H}$.
The goal is to estimate the state $U\ket{d}$ to within trace distance error $\varepsilon$ using only forward queries to $U$.
\end{problem}

\begin{theorem}\label{thm-8100202}
    There is a quantum protocol that uses $n$ forward queries to $U$ and outputs a classical description of a unit vector $\ket{\psi}$. The output $\ket{\psi}$ is random and satisfies
\begin{equation}\label{eq-10201105}
1-\E_{\ket{\psi}}\!\left[|\bra{\psi}U\ket{d}|^2\right]\leq O\!\left(\frac{d^3}{n(n+d^2)}\right).
\end{equation}
\end{theorem}
By \cref{thm-8100202}, we can easily obtain the following corollaries.
\begin{corollary}\label{corollary-10201505}
The inverse-free state estimation problem can be solved using $O(\min\{d^{3/2}/\varepsilon,d/\varepsilon^2\})$ queries.
\end{corollary}
\begin{proof}
Using \cref{thm-8100202} and letting $n=O(\min\{d^{3/2}/\varepsilon,d/\varepsilon^2\})$, we can obtain a random unit vector $\ket{\psi}$ such that 
$1-\E_{\ket{\psi}}\!\left[|\bra{\psi}U\ket{d}|^2\right]\leq \varepsilon^2/100$, which means $\Pr[1-|\bra{\psi} U\ket{d}|^2\leq \varepsilon^2/10 ]\geq 2/3$.
Note that $\|\ketbra{\psi}{\psi}-U\ketbra{d}{d}U^\dag\|_1=2\sqrt{1-|\bra{\psi}U\ket{d}|^2}$ where $\|\cdot\|_1$ denotes the trace norm. Thus the trace distance error is at most $\varepsilon$ with probability at least $2/3$.
\end{proof}
Then, we introduce our result for amplitude estimation.
\begin{problem}[Amplitude estimation]
Let $\mathcal{H}\cong \mathbb{C}^d$ be a $d$-dimensional Hilbert space with an orthonormal basis $\{\ket{1},\ket{2},\ldots,\ket{d}\}$.
Suppose $U$ is an unknown unitary on $\mathcal{H}$ and $\Pi$ is an orthogonal projector onto a subspace of $\mathcal{H}$.
The goal is to estimate the amplitude of $U\ket{d}$ on $\Pi$, i.e. $\sqrt{\bra{d} U^\dag \Pi U\ket{d}}$, to within error $\varepsilon$ using only forward queries to $U$.
\end{problem}
\begin{corollary}
The inverse-free amplitude estimation problem can be solved using $O(\min\{d^{3/2}/\varepsilon,1/\varepsilon^2\})$ queries.
\end{corollary}
\begin{proof}
By \cref{lemma-10201504}, any state estimation protocol with trace distance error $\varepsilon$ is also an amplitude estimation protocol with error $O(\varepsilon)$.
Therefore, \cref{corollary-10201505} combined with the naive $O(1/\varepsilon^2)$ amplitude estimation can imply an $O(\min\{d^{3/2}/\varepsilon,1/\varepsilon^2\})$ query upper bound for the inverse-free amplitude estimation.
\end{proof}

\subsection{Related work}
\paragraph{State estimation with Heisenberg scaling.}
Van Apeldoorn, Cornelissen, Gily\'en and Nannicini~\cite{van2023quantum} developed a family of quantum state estimation protocols based on block encoding techniques~\cite{low2019hamiltonian,chakraborty2019ICALP,van2018improvements} and quantum singular value transformation~\cite{gilyen2019quantum}. They showed that $O(d\log(d)/\varepsilon)$ queries to the state-preparation unitary and its inverse suffice, and established a nearly tight lower bound $\widetilde{\Omega}(d/\varepsilon)$.
This result achieves a quadratic improvement in the dependence on the error parameter, reducing the classical scaling of $\Theta(1/\varepsilon^2)$ (in the sample access model) to $\Theta(1/\varepsilon)$, thereby achieving what is known as Heisenberg scaling.
Moreover, they also provided protocols applicable to estimating a mixed state $\rho$ of rank $r$, assuming that the state-preparation unitary produces a purification of $\rho$. In this case, they showed that $O(dr\polylog(d)/\varepsilon)$ queries are sufficient. However, since their methods are based on quantum singular value transformation, the use of inverse queries appear to be an inherent limitation of their approach.

\paragraph{Power of inverse queries.}
The power of inverse queries has been extensively investigated across a wide range of quantum information tasks. A particularly relevant and directly related question, known as unitary time-reversal, asks how to simulate a single call to the inverse unitary $U^{-1}$ using only forward queries to $U$. It was recently established that this task has query complexity $\Theta(d^2)$~\cite{chen2024quantum,odake2024analytical,chen2025tight}, implying that an exponential overhead is unavoidable when attempting to directly simulate a protocol involving inverses in an inverse-free setting. 
Exponential separations between the inverse-ful and inverse-free models have also been demonstrated in several other contexts, including quantum learning~\cite{cotler2023information, schuster2024random} and preimage finding with an in-place permutation oracle~\cite{fefferman2015quantum}. Notably, amplitude estimation—a fundamental problem that serves as a key subroutine in many quantum algorithms—has been shown to require inverse access.
Tang and Wright~\cite{tang2025amplitude} established a query lower bound $\Omega(\min\{d, 1/\varepsilon^2\})$ in the inverse-free model, thereby separating it from the $O(1/\varepsilon)$ scaling achievable when inverse queries are available~\cite{brassard2000quantum}.

In contrast, for the unitary estimation task~\cite{kahn2007fast,yang2020optimal,haah2023query,grewal2025query}, one can achieves optimal query complexity without using inverse queries.
Kahn~\cite{kahn2007fast} showed an inverse-free unitary estimation protocol with query complexity $O(f(d)/\varepsilon)$, where the estimate achieves an entanglement infidelity error $\varepsilon^2$ and $f(d)$ was left unspecified. 
Notably, in 2020, Yang, Renner, and Chiribella~\cite{yang2020optimal} were able to fully determine the scaling to be $O(d^2/\varepsilon)$. More recently, Haah, Kothari, O'Donnell, and Tang~\cite{haah2023query} proved that the scaling $O(d^2/\varepsilon)$ even holds for diamond-norm error $\varepsilon$, along with a matching lower bound $\Omega(d^2/\varepsilon)$, even when inverses are allowed.
These findings highlight that the power of inverse queries varies substantially across tasks. A general understanding of when and how an advantage exists remains an open question.

\subsection{Overview of techniques}
The estimation protocol is conceptually simple. We apply $U^{\otimes n}$ in parallel to a suitably chosen probe state $\ket{p}\in\mathcal{H}^{\otimes n}$ and then perform an appropriate measurement on the resulting state to extract the desired information.
The design of the probe state and the measurement is based on representation-theoretic techniques, inspired by the works~\cite{kahn2007fast,yang2020optimal}.
However, different from them, we are only interested in the last column of $U$ (i.e., $U\ket{d}$), thus the probe state should be tailored to extract information exclusively from that column.
Guided by this intuition, we construct a family of semistandard Young tableaux $\Gamma_i$ defined in \cref{fig-10260110}, whose corresponding Gelfand–Tsetlin basis vectors $\ket{\Gamma_i}$ possess the desired properties. The probe state $\ket{p}$ is then taken as a superposition of these vectors with carefully chosen coefficients.

We provide here some intuition about the state $\ket{\Gamma_i}$. Consider the subgroup $\mathbb{U}_{d-1}\subset\mathbb{U}_d$ consisting of unitaries that fix the vector $\ket{d}$. By the definition of Gelfand-Tsetlin basis, the action of any $V\in\mathbb{U}_{d-1}$ on $\ket{\Gamma_i}$ corresponds to the action of $V$ in the $1$-dimensional $\mathbb{U}_{d-1}$-representation labeled by the Young diagram of $d-1$ rows and $L$ columns, which simply multiplies $\ket{\Gamma_i}$ by a global phase $\det(V)^L$. 
This means, for any $U,V\in\mathbb{U}_d$ such that $U\ket{d}=V\ket{d}$, we have $U^\dag V\in\mathbb{U}_{d-1}$ and so the actions of $U$ and $V$ on $\ket{\Gamma_i}$ differ only by the phase $\det(U^\dag V)^L$.
Therefore, $\ket{\Gamma_i}$ probes information solely from the last column.

One can also interpret it using dual representations. 
For convenience, let us assume the unitary $U$ has determinant $1$ (i.e., $U\in \mathbb{SU}_d$). 
When restricting to $\mathbb{SU}_d$-representations, we can discard the first $i$ full-length columns of $\Gamma_i$.
Then, let $\lambda$ be the Young diagram with $d-1$ rows and $1$ column and let $S_i$ be the semistandard Young tableau obtained from $\lambda$ by filling it with the integers $\{1,\ldots,d\}\setminus i$.
There is a natural isomorphism between $\mathbb{SU}_d$-representations $\mathcal{Q}_{\lambda}^d \rightarrow \mathcal{Q}_{-\Box}^d: \ket{S_i} \mapsto (-1)^{d-i}\ket{\underline{i}}$, where $\mathcal{Q}_{-\Box}^d$ is the dual representation of $\mathcal{Q}_{\Box}^d\cong\mathcal{H}$, $\ket{\underline{i}}\in\mathcal{Q}_{-\Box}^d$ is the covector corresponding to $\ket{i}\in\mathcal{Q}_{\Box}^d$. Thus, under these isomorphisms, we can interpret $\Gamma_i$ as
\[\includegraphics[width=0.8\linewidth]{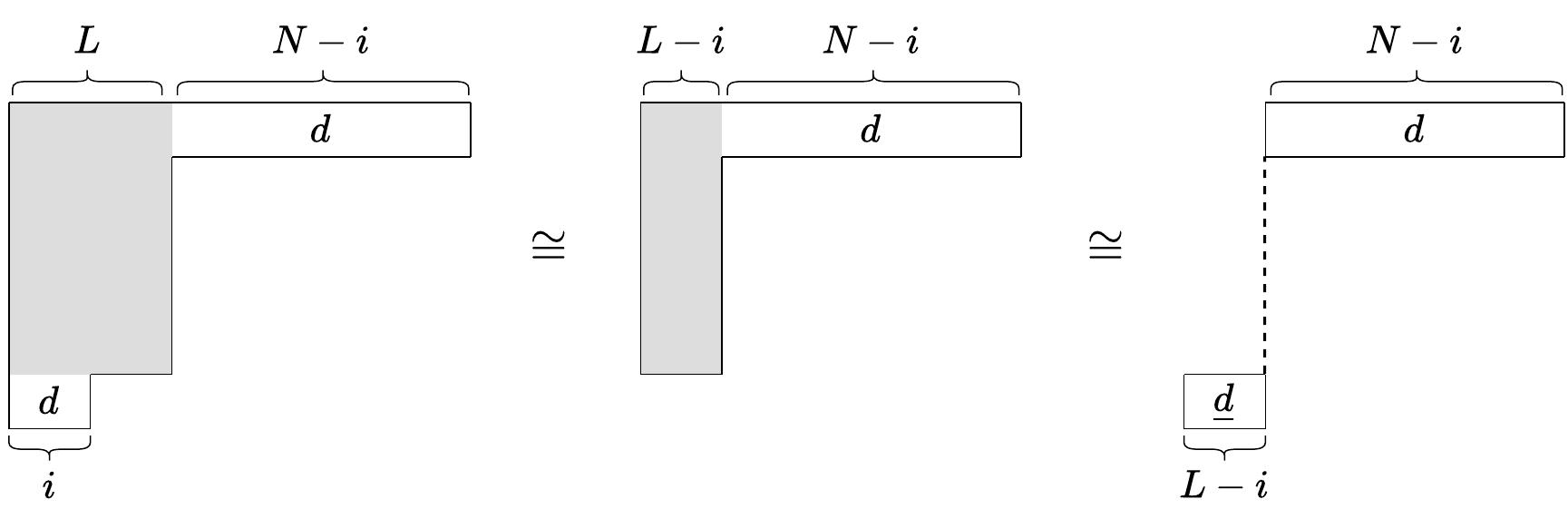},\]
where the gray region contains multiple $S_d$ arranged horizontally, the white region labeled by $d$ ($\underline{d}$) contains boxes filling with $d$ ($\underline{d}$), the final tableau represents a mixed semistandard Young tableau~\cite{grinko2023gelfand,nguyen2023mixed}.
This means $\ket{\Gamma_i}$ can be viewed as an irreducible components in 
\[\ket{\underline{d}}^{\otimes L-i}\otimes \ket{d}^{\otimes N-i}\in (\mathcal{Q}_{-\Box}^d)^{\otimes L-i}\otimes (\mathcal{Q}_{\Box}^d)^{\otimes N-i}.\]
Note that $U$ acts as $U^*$ on $\mathcal{Q}_{-\Box}^d$. Thus, we can see that the action of $U$ on $\ket{\Gamma_i}$ depends only on $U\ket{d}$.

Then, we use the pretty good measurement (PGM)~\cite{hausladen1994pretty,belavkin1975optimal,belavkin1975optimum,barnum2002reversing} to extract information from the state $U^{\otimes n}\ket{p}$. Specifically, we consider it as a state discrimination problem where our goal is to identify the label $U$ from the ensemble $\{(\textup{d} V, V^{\otimes n}\ket{p})\}_{V\in\mathbb{U}_d}$ with $\textup{d} V$ denoting the Haar measure. Let $\widehat{U}$ denote the measurement outcome. We set $\widehat{U}\ket{d}$ as the estimate of $U\ket{d}$, and show that the expected infidelity of this estimator can be upper bounded by $O(d^3/(n^2+nd^2))$, thereby establishing \cref{thm-8100202}.

\subsection{Discussion}\label{sec-10271551}
{\renewcommand{\arraystretch}{1.5}
\begin{table}[ht]
\centering
\begin{tabular}{|c|cc|}
\hline
                   & \multicolumn{1}{c|}{Inverse-free} & Inverse-ful \\ \hline
Unitary estimation                  & \multicolumn{2}{c|}{$\Theta(d^2/\varepsilon)$~\cite{haah2023query}}                           \\ \hline
\multirow{2}{*}{State estimation} & \multicolumn{1}{c|}{$O(\min\{d^{3/2}/\varepsilon,d/\varepsilon^2\})$ {(This work)}}             &    \multirow{2}{*}{$\widetilde{\Theta}(d/\varepsilon)$~\cite{van2023quantum}}       \\
                   & \multicolumn{1}{c|}{$\widetilde{\Omega}(d/\varepsilon)$~\cite{van2023quantum}}            &           \\ \hline
\multirow{2}{*}{Amplitude estimation} & \multicolumn{1}{c|}{$O(\min\{d^{3/2}/\varepsilon,1/\varepsilon^2\})$ {(This work)}}             &    \multirow{2}{*}{$\Theta(1/\varepsilon)$~\cite{brassard2000quantum}}         \\
                   & \multicolumn{1}{c|}{$\Omega(\min\{d,1/\varepsilon^2\})$~\cite{tang2025amplitude}}             &             \\ \hline
\end{tabular}
\caption{Query upper and lower bounds for unitary, state and amplitude estimation in inverse-free and inverse-ful settings.} \label{tab-10211545}
\end{table}
}

\cref{tab-10211545} summarizes the existing results of unitary, state and amplitude estimation in both inverse-free and inverse-ful settings.
In this table, we conjecture that our upper bound for inverse-free state estimation is tight. We also conjecture that the query complexity of inverse-free amplitude estimation is $\Theta(\min\{d/\varepsilon,1/\varepsilon^2\})$.
Let us consider the additional overhead if one aims to achieve Heisenberg scaling but removes inverse access.
If our conjectures hold, then the corresponding overheads are: $\Theta(1)$ for unitary estimation; $\Theta(\sqrt{d})$ for state estimation; and $\Theta(d)$ for amplitude estimation.
These might suggest that, rather than exhibiting a sharp transition between estimating all features and estimating a single desired feature, the benefit of inverse access may increase smoothly as the number of features to be estimated decreases.

\section{Preliminaries}
\subsection{Young diagrams and Young tableaux}
A \textit{Young diagram} \(\lambda\) consisting of \(n\) boxes and \(\ell\) rows is a partition \((\lambda_1,\ldots,\lambda_\ell)\) of \(n\) such that \(\sum_{i=1}^\ell \lambda_i=n\) and \(\lambda_1\geq\cdots \geq \lambda_\ell> 0\). The $i$-th rows consists of $\lambda_i$ boxes. 
By convention, the Young diagram is drawn with left-justified rows, arranged from top to bottom.
For example, the Young diagram \(\lambda=(4,3,1)\) is drawn as:
\[\vcenter{\hbox{\scalebox{0.9}{\begin{ytableau}~ &~&~&~\\~&~&~\\~\end{ytableau}}}}.\]
We use $\ell(\lambda)$ to denote the number of rows of $\lambda$.
We use \(\lambda\vdash n\) to denote that \(\lambda\) is a Young diagram with \(n\) boxes and use $\lambda\vdash_d n$ to denote that $\lambda\vdash n$ and $\ell(\lambda)\leq d$.

Let $\mu,\lambda$ be Young diagrams. We write $\mu\nearrow\lambda$ or $\lambda\searrow\mu$ if $\lambda$ can be obtained from $\mu$ by adding one box.
We also write $\lambda\succsim\mu$ or $\mu\precsim \lambda$ if  $\ell(\mu)\leq \ell(\lambda)=:\ell$ and 
\[\lambda_1\geq \mu_1\geq \lambda_2\geq \mu_2\geq\cdots\geq \lambda_\ell\geq \mu_\ell,\]
where we define $\mu_i=0$ if $i>\ell(\mu)$. 
Equivalently, this means the row-lengths of $\lambda$ and $\mu$ interlace and we say $\mu$ interlaces $\lambda$. An example is $(4,3,1)\succsim (3,3,1)\succsim (3,1)\succsim (2)$:
\begin{equation}\label{eq-10221712}
\vcenter{\hbox{\scalebox{0.9}{\begin{ytableau}~ &~&~&~\\~&~&~\\~\end{ytableau}}}}
\quad\succsim\quad
\vcenter{\hbox{\scalebox{0.9}{\begin{ytableau}~ &~&~\\~&~&~\\~\end{ytableau}}}}
\quad\succsim\quad
\vcenter{\hbox{\scalebox{0.9}{\begin{ytableau}~ &~&~\\~\end{ytableau}}}}
\quad\succsim\quad
\vcenter{\hbox{\scalebox{0.9}{\begin{ytableau}~ &~\end{ytableau}}}}.
\end{equation}

A \textit{semistandard Young tableau} of shape \(\lambda\vdash n\) and alphabet $[d]$ is a filling of $\lambda$ using the integers from $[d]$, such that the entries weakly increase from left to right in each row and
strictly increase from top to bottom in each column. An example of semistandard Young tableau of shape $(4,3,1)$ and alphebet $[4]$ is
\begin{equation}\label{eq-10221716}
\vcenter{\hbox{\scalebox{0.9}{\begin{ytableau}1&1&2&4\\2&3&3\\3\end{ytableau}}}}.
\end{equation}
We use $\textup{SSYT}(\lambda,d)$ to denote the set of semistandard Young tableaux of shape $\lambda$ with alphabet $[d]$. We also use $\Sh(S)$ to denote the shape of the semistandard Young tableau $S$.

\subsection{Gelfand-Tsetlin basis}
Let $\mathcal{H}\cong\mathbb{C}^d$ be a $d$-dimensional Hilbert space with an orthonormal basis $\{\ket{1},\ket{2},\ldots,\ket{d}\}$.
Let $\mathbb{U}_d$ denote the group of $d\times d$ unitary matrices on $\mathcal{H}$ and $\mathfrak{S}_n$ denote the group of permutations of set $[n]=\{1,2,\ldots,n\}$.

The space $\mathcal{H}^{\otimes n}$ is a representation of the symmetric group \(\mathfrak{S}_n\) and the unitary group \(\mathbb{U}_d\).
The unitary group $\mathbb{U}_d$ acts by simultaneous ``rotation'' as \(U^{\otimes n}\) for any \(U\in \mathbb{U}_d\) and the symmetric group $\mathfrak{S}_n$ acts by permuting tensor factors:
\begin{equation*}
P(\pi)\ket{\psi_1}\cdots\ket{\psi_n}=\ket{\psi_{\pi^{-1}(1)}}\cdots\ket{\psi_{\pi^{-1}(n)}},
\end{equation*}
for any \(\pi\in\mathfrak{S}_n\). 
By the Schur-Weyl duality (see, e.g., \cite{fulton2013representation}), we have the following decomposition 
\begin{equation}\label{eq-10231550}
\mathcal{H}^{\otimes n}\stackrel{\mathbb{U}_d\times\mathfrak{S}_n}{\cong} \bigoplus_{\lambda\vdash_d n}\mathcal{Q}^d_\lambda\otimes\mathcal{P}_\lambda.
\end{equation}
Here, we use $\mathcal{Q}_\lambda^d$ and $\mathcal{P}_\lambda$ to denote the irreducible representations of $\mathbb{U}_d$ and $\mathfrak{S}_n$ corresponding to the Young diagram $\lambda$.
We remark that in this paper, without loss of generality, all representations are unitary representations and the homomorphisms between representations preserve inner product.
We use $\texttt{q}_{\lambda}(U)$ to denote the action of $U$ on $\mathcal{Q}_\lambda^d$.
The dimension of $\mathcal{Q}_\lambda^d$ can be computed using the Weyl dimension formula (see, e.g., \cite{etingof2011introduction})
\begin{equation}\label{eq-10271256}
\dim(\mathcal{Q}_\lambda^d)=\prod_{1\leq i< j\leq d}\frac{\lambda_i-i - \lambda_j+j}{j-i}.
\end{equation}

\paragraph{Branching rule of $\mathbb{U}_d$.}
We now look deeper into the structure of $\mathcal{Q}_\lambda^d$. Consider the following chain of embeddings of groups:
\[\mathbb{U}_1\rightarrow\mathbb{U}_2\rightarrow\cdots\rightarrow\mathbb{U}_d,\]
where the embedding $\mathbb{U}_i\rightarrow \mathbb{U}_{i+1}$ is given by
\[U\mapsto \begin{bmatrix}
U & 0 \\
0 & 1
\end{bmatrix},\]
for any $U\in\mathbb{U}_i$.
Let $\lambda\vdash_i n$ be a Young diagram of $n$ boxes with at most $i$ rows and $\mathcal{Q}^i_\lambda$ be an irreducible representation of $\mathbb{U}_i$. Then, $\mathcal{Q}^i_\lambda$ can be viewed as a representation of $\mathbb{U}_{i-1}$, which is called the restriction of $\mathcal{Q}^i_\lambda$ to $\mathbb{U}_{i-1}$ and is denoted by $\textup{Res}^{\mathbb{U}_{i}}_{\mathbb{U}_{i-1}}(\mathcal{Q}^i_\lambda)$. Note that $\textup{Res}^{\mathbb{U}_{i}}_{\mathbb{U}_{i-1}}(\mathcal{Q}^i_\lambda)$ also decomposes as a direct sum of irreducible representations of $\mathbb{U}_{i-1}$. Specifically, we have the following branching rule (see, e.g., \cite{goodman2009symmetry})
\begin{equation}\label{eq-10222015}
\textup{Res}_{\mathbb{U}_{i-1}}^{\mathbb{U}_i}(\mathcal{Q}_\lambda^i)\stackrel{\mathbb{U}_{i-1}}{\cong}\bigoplus_{\substack{\mu:\mu\precsim\lambda\\ \ell(\mu)\leq i-1}}\mathcal{Q}^{i-1}_\mu.
\end{equation}
Thus, we can see that $\mathcal{Q}^d_\lambda$ decomposes into a direct sum of pairwise non-equivalent irreducible representations and they are pairwise orthogonal.
By iterating this process (i.e., taking further restrictions on $\mathcal{Q}^d_\mu$), we finally obtain a direct sum of the one-dimensional irreducible representations of $\mathbb{U}_1$ (since $\mathbb{U}_1$ is Abelian thus has only $1$-dimensional irreducible representations). Therefore, this process determines an orthonormal basis (up to phases) of $\mathcal{Q}^d_\lambda$ and each basis vector can be parameterized by its branching path:
\begin{equation}\label{eq-10222040}
\lambda^{(1)}\rightarrow\lambda^{(2)}\rightarrow\cdots\rightarrow\lambda^{(d)},
\end{equation}
where $\lambda^{(d)}=\lambda$, $\lambda^{(i-1)}\precsim \lambda^{(i)}$ and $\ell(\lambda^{(i)})\leq i$.
This basis is called the \textit{Gelfand-Tsetlin basis} of $\mathcal{Q}_\lambda^d$.

Furthermore, we can see that the path in \cref{eq-10222040} uniquely corresponds to a semistandard Young tableau $S$ of shape $\lambda=\lambda^{(d)}$, where the integer $i$ is assigned to the boxes in $\lambda^{(i)}\setminus\lambda^{(i-1)}$, in which $\lambda^{(i)}\setminus\lambda^{(i-1)}$ denotes the shape by removing from the shape of $\lambda^{(i)}$ all the boxes belonging to $\lambda^{(i-1)}$. 
For example, the chain in \cref{eq-10221712} corresponds to the semistandard Young tableau in \cref{eq-10221716}.
By this mean, we can identify a semistandard Young tableau $S$ with alphabet $[d]$ uniquely by a length-$d$ interlacing sequence of Young diagrams as in \cref{eq-10222040}. 
Thus each vector in the Gelfand-Tsetlin basis of $\mathcal{Q}_\lambda^d$ is labeled by a semistandard Young tableau $S$ of shape $\lambda$ with alphabet $[d]$, and we will use $\ket{S}$ to denote such vector.

\subsection{Clebsch-Gordan coefficients}\label{sec-10241631}
Let us consider the tensor product of the irreducible representations $\mathcal{Q}_\lambda^d$ and $\mathcal{Q}_{\Box}^d$, where $\Box$ denotes the Young diagram with only one box and thus $\mathcal{Q}_{\Box}^d\cong \mathcal{H}$ is just the defining representation of $\mathbb{U}_d$. 
This tensor product representation is generally reducible and decomposes as
\[\mathcal{Q}_\lambda^d \otimes \mathcal{H}\stackrel{\mathbb{U}_d}{\cong}\bigoplus_{\substack{\mu:\lambda\nearrow\mu\\ \ell(\mu)\leq d}}\mathcal{Q}_\mu^d,\]
due to the Pieri's rule (see, e.g., \cite{goodman2009symmetry}). Under this isomorphism, we can write the tensor product of vectors $\ket{S}\in\mathcal{Q}^d_\lambda$ and $\ket{i}\in \mathcal{H}$, where $S\in\textup{SSYT}(\lambda,d)$ and $i\in[d]$, as
\[\ket{S}\otimes \ket{i}=\sum_{\substack{\mu:\lambda\nearrow\mu\\ \ell(\mu)\leq d}} \sum_{T\in \textup{SSYT}(\mu,d)} C^{S,i}_{T}\ket{T},\]
where $C^{S,i}_{T}$ is called the Clebsch-Gordan coefficients which can be further expressed as a product of reduced Wigner coefficients~\cite{harrow2005applications,bacon2007quantum,bacon2006efficient} or scalar factors~\cite{Vilenkin1992}.
When $i=d$, the decomposition has a simpler form:
\begin{equation}\label{eq-10241543}
\ket{S}\otimes \ket{d}=\sum_{k=1}^d C^{S}_{k} \ket{S_{k}},
\end{equation}
where $S_{k}$ denotes the semistandard Young tableau obtained from $S$ by adding a box filling with $d$ to the $k$-th row.\footnote{Note that $S_{k}$ may not be a valid semistandard Young tableau, but in such case $C_{k}^{S}=0$.} Let $\lambda^{(1)}\rightarrow\cdots\rightarrow\lambda^{(d)}=\lambda$ be the interlacing sequence of Young diagrams corresponding to $S$, then $C_{k}^S$ can be expressed as (see, e.g., Equation 3 in \cite[Sec. 18.2.10]{Vilenkin1992}, and also \cite[Sec. 10]{pelecanos2025debiased}, \cite[Apdx. A]{grinko2023gelfand}):
\begin{equation}\label{eq-10241544}
 C_{k}^S=\left|\frac{ \prod_{j=1}^{d-1}(\lambda^{(d-1)}_j-j-\lambda_k+k-1)}{ \prod_{\substack{j=1,j\neq k}}^d(\lambda_j-j-\lambda_k+k)}\right|^{1/2}.
\end{equation}

\section{Estimation protocol}
In this section, we introduce the estimation protocol.
Let $\mathcal{H}\cong\mathbb{C}^d$ be a $d$-dimensional Hilbert space and $U$ be an unknown unitary on $\mathcal{H}$.
We apply $U^{\otimes n}$ on a state $\ket{p}\in\mathcal{H}^{\otimes n}$, called the probe state. Then we perform a POVM on $U^{\otimes n}\ket{p}$, obtaining a ``partial'' estimate $\widehat{U}$ of $U$. Finally, we set $\widehat{U}\ket{d}$ as our estimate of $U\ket{d}$.

\subsection{A coherent system on \texorpdfstring{$\Gamma$}{Gamma}-shaped tableaux}
Let $n$ be the number of queries. Without loss of generality, we can assume $n$ is a multiple of $2d$.
Let 
\[L\coloneqq \frac{n}{2d} \quad\textup{and}\quad N\coloneqq (d+1)L=\frac{(d+1)n}{2d}.\]
Then, we define a family of semistandard Young tableaux that have ``$\Gamma$''-like shape.
\begin{definition}[$\Gamma$-shaped tableaux]
For any integer $i$ such that $0\leq i\leq L$, let $\Gamma_{i}$ and $\Gamma_{i}^+$ be the semistandard Young tableaux defined in \cref{fig-880350}.
\end{definition}
\begin{figure}[ht]
    \centering
    \begin{subfigure}[b]{0.4\linewidth}
        \centering
    \includegraphics[width=1.0\linewidth]{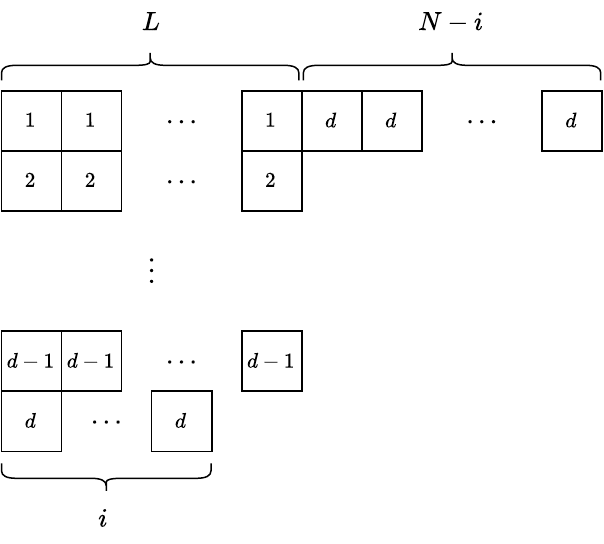}
    \caption{$\Gamma_i$}\label{fig-10260110}
    \end{subfigure}
    \hspace{10mm}
    \begin{subfigure}[b]{0.4\linewidth}
    \centering
    \includegraphics[width=1.0\linewidth]{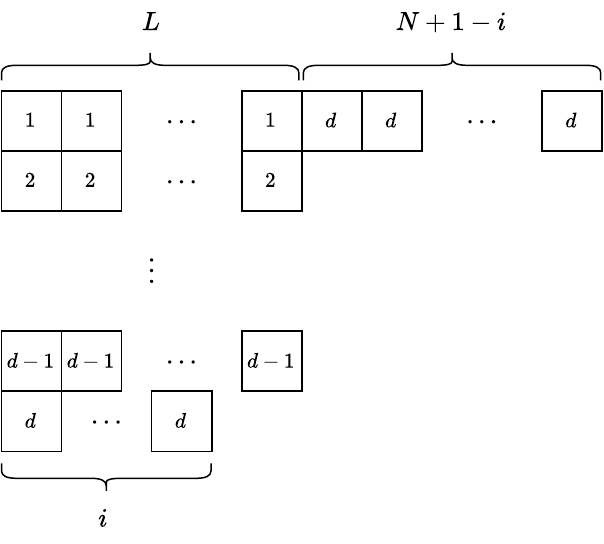}
    \caption{$\Gamma_i^+$}
    \end{subfigure}
    \caption{Definition of the semistandard Young tableaux $\Gamma_i$ and $\Gamma_i^+$.}
    \label{fig-880350}
\end{figure}
It is easy to see that $\Gamma_i$ contains $(d-1)L+N=n$ boxes and $\Gamma_i^+$ contains $(d-1)L+N+1=n+1$ boxes.
Let $\gamma_i\coloneqq\Sh(\Gamma_i)$ and $\gamma_{i}^+=\Sh(\Gamma_i^+)$ be the shapes of $\Gamma_i$ and $\Gamma_i^+$.
Then, $\ket{\Gamma_i}$ and $\ket{\Gamma_i^+}$ are two vectors in the Gelfand-Tsetlin bases of $\mathcal{Q}_{\gamma_i}^d$ and $\mathcal{Q}_{\gamma_i^+}^d$, respectively.

We note that these tableaux are correlated if we consider the tensor product of $\ket{\Gamma_i}$ and $\ket{d}$.
By the Clebsch-Gordan decomposition (see \cref{sec-10241631}), $\ket{\Gamma_i}\otimes \ket{d}$ decomposes into a sum of vectors corresponding to the semistandard Young tableaux obtained from $\Gamma_i$ by adding one box filling with $d$. It is easy to see that such box can only be added at the first and the $d$-th rows.
More specifically, from \cref{eq-10241543} and \cref{eq-10241544}, we have
\begin{align}
\ket{\Gamma_i}\otimes \ket{d}&=\left|\frac{\prod_{1\leq j\leq d-1}(N-i+j)}{(L+N+d-2i-1)\prod_{2\leq j \leq d-1}(N-i+j-1)}\right|^{1/2}\ket{\Gamma_{i}^+} \nonumber \\
&\qquad\qquad\qquad\qquad\qquad\qquad+ \left|\frac{\prod_{1\leq j\leq d-1}(L+d-i-j-1)}{(L+N+d-2i-1)\prod_{2\leq j\leq d-1}(L+d-i-j)}\right|^{1/2}\ket{\Gamma_{i+1}^+}\nonumber\\
&=\sqrt{\frac{N+d-i-1}{L+N+d-2i-1}}\ket{\Gamma_{i}^+}+\sqrt{\frac{L-i}{L+N+d-2i-1}}\ket{\Gamma_{i+1}^+}.\label{eq-830204}
\end{align}
We define
\begin{equation}\label{eq-872206}
\alpha_i\coloneqq \frac{N+d-i-1}{L+N+d-2i-1},\quad\quad \beta_i\coloneqq \frac{L-i}{L+N+d-2i-1}.
\end{equation}
Then, \cref{eq-830204} can be written as
\begin{equation}\label{eq-831608}
\ket{\Gamma_i}\otimes \ket{d}=\sqrt{\alpha_i}\ket{\Gamma_i^+}+\sqrt{\beta_i}\ket{\Gamma^+_{i+1}},
\end{equation}
which is also shown in \cref{fig-10241656}.
\begin{figure}[ht]
    \centering
    \includegraphics[width=0.9\linewidth]{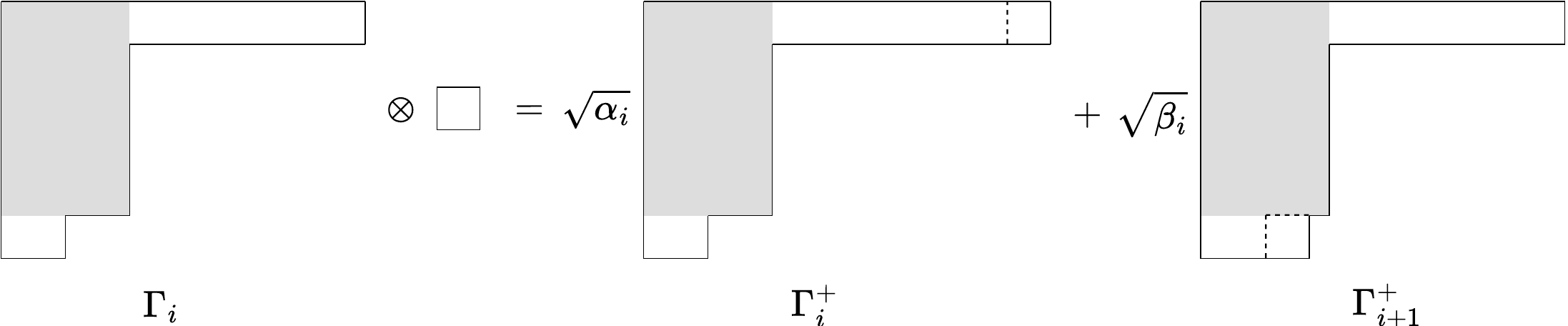}
    \caption{The decomposition of $\ket{\Gamma_i}\otimes \ket{d}$}
    \label{fig-10241656}
\end{figure}

\subsection{Probe state}
Then, we define the probe state $\ket{p}$.
For each $0\leq i\leq L$, fix an arbitrary unit vector $\ket{\bot_{\gamma_i}}\in\mathcal{P}_{\gamma_i}$, where $\mathcal{P}_{\gamma_i}$ the irreducible representation of $\mathfrak{S}_n$ labeled by $\gamma_i$. The state $\ket{\bot_{\gamma_i}}$ will be used in our definition of the probe state and measurement, while the final results are independent of the choice of $\ket{\bot_{\gamma_i}}$.
\begin{definition}[The probe state]
\begin{equation}\label{eq-882202}
\ket{p}\coloneqq \sum_{i=0}^L f_{i} \ket{\Gamma_{i}}\ket{\bot_{\gamma_i}}\in\mathcal{H}^{\otimes n},
\end{equation}
where $f_i$ are real numbers such that $\sum_{i=0}^L f_i^2=1$, and will be determined in \cref{sec-10251701}.
\end{definition}

\subsection{Pretty good measurement}
By applying $U^{\otimes n}$ to the probe state $\ket{p}$, we obtain a state $U^{\otimes n}\ket{p}$ that encodes some information of $U$. 
To extract the information from this state, we view it as a state discrimination problem.
That is, suppose we are given an unknown state drawn from the ensemble 
$\{(\textup{d} V, V^{\otimes n}\ket{p})\}_{V\in\mathbb{U}_d}$, where $\textup{d} V$ is the Haar measure, and we want to determine which state we have received (i.e., to identify the corresponding label $V$).
Our idea is to perform the so-called ``pretty good measurement'' (PGM)~\cite{hausladen1994pretty,belavkin1975optimal,belavkin1975optimum,barnum2002reversing}.

The PGM with respect to this ensemble is defined as a continuous POVM 
\[\{E_V\}_{V\in\mathbb{U}_d},\]
where 
\[E_V\coloneqq \textup{d}V \, E^{-\frac{1}{2}} V^{\dag\otimes n}\ketbra{p}{p}V^{\otimes n} E^{-\frac{1}{2}},\]
in which $\textup{d}V$ is the Haar measure, $E^{-\frac{1}{2}}$ is the pseudo-inverse of $E^{\frac{1}{2}}$ and $E$ is the average of the ensemble:
\begin{align}
E&\coloneqq \int \textup{d}V \, V^{\otimes n}\ketbra{p}{p} V^{\dag\otimes n} \nonumber\\
&=\int \textup{d}V \, \sum_{i,j=0}^L f_i f_j \texttt{q}_{\gamma_i}(V) \ketbra{\Gamma_i}{\Gamma_j}\texttt{q}_{\gamma_j}(V^\dag) \otimes \ketbra{\bot_{\gamma_i}}{\bot_{\gamma_j}} \nonumber \\
&=\sum_{i=0}^L \frac{f_i^2}{\dim(\mathcal{Q}_{\gamma_i}^d)} I_{\mathcal{Q}_{\gamma_i}^d} \otimes \ketbra{\bot_{\gamma_i}}{\bot_{\gamma_i}},\label{eq-10242214}
\end{align}
where \cref{eq-10242214} is by the Schur's lemma.
Thus, we can see that
\[E_V=\textup{d}V \, \ketbra{\eta_V}{\eta_V},\]
where 
\begin{equation}\label{eq-882206}
\ket{\eta_{V}}\coloneqq \sum_{i=0}^L \sqrt{\dim(\mathcal{Q}_{\gamma_i}^d)}\texttt{q}_{\gamma_i}(V) \ket{\Gamma_{i}}\ket{\bot_{\gamma_i}}.
\end{equation}

\subsection{The protocol}
Then, the estimation protocol is shown in Protocol~\ref{alg-882158}.

\begin{algorithm}[H]
\caption{\raggedright Estimating $U\ket{d}$}\label{alg-882158}
    \begin{algorithmic}[1]
    \Require $n$ forward queries to $U$.
    \Ensure An estimate of $U\ket{d}$.
    \State Prepare the probe state $\ket{p}$. \Comment{see \cref{eq-882202}}
    \State Apply $U^{\otimes n}$ on $\ket{p}$.
    \State Apply the POVM $\{\textup{d}V\, \ketbra{\eta_V}{\eta_V}\}_{V\in\mathbb{U}_d}$ on $U^{\otimes n}\ket{p}$ and obtain an outcome $\widehat{U}$. \Comment{see \cref{eq-882206}}
    \State\textbf{Return} $\widehat{U}\ket{d}$.
    \end{algorithmic}
\end{algorithm}

\section{Bounding the error}\label{sec-10241855}
In this section, we bound the error of Protocol~\ref{alg-882158}.

For $0\leq i\leq L$, we define 
\begin{equation}\label{eq-10270121}
\begin{split}
x_i&\coloneqq \sqrt{\frac{(N+d-i-1)(N-i+1)}{(L+N+d-2i-1)(L+N+d-2i)}},\\
y_{i}&\coloneqq \sqrt{\frac{(L-i+1)(L+d-i-1)}{(L+N+d-2i+1)(L+N+d-2i)}}.
\end{split}
\end{equation}
Then, we have the following lemma.
\begin{lemma}
Let $\widehat{U}\ket{d}$ be the output of Protocol~\ref{alg-882158}.
The expectation of fidelity between $\widehat{U}\ket{d}$ and $U\ket{d}$ is
\begin{equation}\label{eq-872234}
\E_{\widehat{U}}\!\left[\left|\bra{d}\widehat{U}^\dag U\ket{d} \right|^2\right]=f_0^2x_0^2+\sum_{i=1}^L (f_{i}x_i+f_{i-1}y_{i})^2.
\end{equation}
\end{lemma}
\begin{proof}
Consider the probability $\Pr(\widehat{U}; U)$ of getting the measurement outcome $\widehat{U}$ when performing the PGM on the state $U^{\otimes n}\ket{p}$:
\begin{align}
\Pr(\widehat{U}; U)&=\textup{d} \widehat{U} \,  |\bra{\eta_{\widehat{U}}}U^{\otimes n}\ket{p}|^2\nonumber\\
&= \textup{d}\widehat{U} \left|\sum_{i=0}^L f_i\sqrt{\dim(\mathcal{Q}_{\gamma_i}^d)}\cdot \bra{\Gamma_{i}}\texttt{q}_{\gamma_i}\!\left(\widehat{U}^\dag U\right)\ket{\Gamma_{i}}\right|^2.\nonumber
\end{align}
Upon getting measurement outcome $\widehat{U}$, we set $\widehat{U}\ket{d}$ as the estimate of $U\ket{d}$. Thus the expectation of (squared) fidelity is
\begin{align}
\E_{\widehat{U}}\!\left[\left|\bra{d}\widehat{U}^\dag U\ket{d} \right|^2\right] &=\int \textup{d}\widehat{U} \left|\sum_{i=0}^L f_i\sqrt{\dim(\mathcal{Q}^d_{\gamma_i})}\cdot \bra{\Gamma_{i}}\texttt{q}_{\gamma_i}\!\left(\widehat{U}^\dag U\right)\ket{\Gamma_{i}}\right|^2\cdot \left|\bra{d}\widehat{U}^\dag U\ket{d}\right|^2\nonumber\\
&=\int \textup{d}\widehat{U} \left|\sum_{i=0}^L f_i \sqrt{\dim(\mathcal{Q}^d_{\gamma_i})}\cdot \bra{\Gamma_{i}}\bra{d} \texttt{q}_{\gamma_i}\!\left(\widehat{U}^\dag U\right) \otimes (\widehat{U}^\dag U) \ket{\Gamma_{i}}\ket{d} \right|^2 \nonumber\\
&=\int \textup{d}\widehat{U} \left|\sum_{i=0}^L f_i \sqrt{\dim(\mathcal{Q}^d_{\gamma_i})}\cdot \left( \alpha_i \bra{\Gamma_i^+} \texttt{q}_{\gamma_i^+}\!\left(\widehat{U}^\dag U\right) \ket{\Gamma_i^+}+\beta_i \bra{\Gamma_{i+1}^+} \texttt{q}_{\gamma_{i+1}^+}\!\left(\widehat{U}^\dag U\right) \ket{\Gamma_{i+1}^+} \right)\right|^2\label{eq-831607}\\
\begin{split}
&=\int \textup{d}\widehat{U} \Bigg| f_0\alpha_0\sqrt{\dim(\mathcal{Q}^d_{\gamma_0})}\bra{\Gamma_0^+}\texttt{q}_{\gamma_0^+}\!\left(\widehat{U}^\dag U\right)\ket{\Gamma_0^+} \\ 
&\qquad\qquad\quad+\sum_{i=1}^{L}\bigg(f_i\alpha_i\sqrt{\dim(\mathcal{Q}^d_{\gamma_i})}+f_{i-1}\beta_{i-1}\sqrt{\dim(\mathcal{Q}^d_{\gamma_{i-1}})}\bigg)\bra{\Gamma_i^+}\texttt{q}_{\gamma_i^+}\!\left(\widehat{U}^\dag U\right)\ket{\Gamma_i^+} \Bigg|^2 \label{eq-841816}
\end{split}\\
&=f_0^2\alpha_0^2\frac{\dim(\mathcal{Q}^d_{\gamma_0})}{\dim(\mathcal{Q}^d_{\gamma_0^+})}+ \sum_{i=1}^{L} \left(f_i\alpha_i\sqrt{\frac{\dim(\mathcal{Q}^d_{\gamma_i})}{\dim(\mathcal{Q}^d_{\gamma_i^+})}} + f_{i-1}\beta_{i-1}\sqrt{\frac{\dim(\mathcal{Q}^d_{\gamma_{i-1}})}{\dim(\mathcal{Q}^d_{\gamma_i^+})}}\right)^2,\label{eq-831606}
\end{align}
where in \cref{eq-831607} we use \cref{eq-831608}, in \cref{eq-841816} we use the fact that $\beta_L=0$, \cref{eq-831606} is by using the fact
\begin{align}
\int \textup{d}\widehat{U}\,\, q_{\gamma_i^+}\!\left(\widehat{U}^\dag U\right)^\dag\ketbra{\Gamma_i^+}{\Gamma_j^+}q_{\gamma_j^+}\!\left(\widehat{U}^\dag U\right)&= \int \textup{d}\widehat{U}\,\, q_{\gamma_i^+}\!\left(\widehat{U}\right)^\dag\ketbra{\Gamma_i^+}{\Gamma_j^+}q_{\gamma_j^+}\!\left(\widehat{U}\right) \nonumber\\
&=\mathbbm{1}_{i=j} \frac{1}{\dim(\mathcal{Q}^d_{\gamma_i^+})} \cdot I_{\mathcal{Q}^d_{\gamma_i^+}},\nonumber
\end{align}
due to the invariance of Haar measure and Schur's lemma. 
Then, from \cref{fact-872212} and \cref{eq-831606}, we can see that
\[\E_{\widehat{U}}\!\left[\left|\bra{d}\widehat{U}^\dag U\ket{d} \right|^2\right]=f_0^2x_0^2+\sum_{i=1}^L (f_{i}x_i+f_{i-1}y_{i})^2.\]
\end{proof}

\subsection{Instantiating \texorpdfstring{$f_i$}{fi}}\label{sec-10251701}
In this section, we will give an explicit form of $f_i$ such that the expectation of fidelity (see \cref{eq-872234}) is as large as possible.
For convenience, we discard the constraint $\sum_{i=0}^L f_i^2=1$ and consider the homogeneous form of the RHS of \cref{eq-872234}:
\begin{equation}\label{eq-831743}
\frac{\sum_{i=0}^L(f_i x_i+f_{i-1}y_i)^2}{\sum_{i=0}^L f_i^2},
\end{equation}
where we set $f_{-1}=0$.
Therefore, we need to find a sequence of real numbers $\{f_i\}_{i=-1}^L$ such that $f_{-1}=0$ and \cref{eq-831743} is as large as possible, or equivalently, $1-\cref{eq-831743}$ is as small as possible.

For $-1\leq i\leq L$, we set
\begin{equation}\label{eq-882258}
f_i
=g_i\cdot \sqrt{(L+N+d-2i-1)\prod_{j=2}^{d-1}(N+j-i-1)\prod_{j=2}^{d-1}(L+d-j-i)},
\end{equation}
where $g_{-1}=0$ and for $0\leq i\leq L$,
\begin{equation}\label{eq-882338}
g_i= (i+1)(L+N+d-i).
\end{equation}

Then, we have the following results, where the proofs are based on repeatedly using \cref{fact-882356} and are deferred.
\begin{lemma}\label{lemma-882333}
Using the instantiation of $f_i$ in \cref{eq-882258}, we can show that $1 - \cref{eq-831743}$ equals
\begin{equation}\label{eq-880122}
\frac{\sum_{i=0}^L \frac{(g_i-g_{i-1})^2}{L+N+d-2i} \prod_{j=1}^{d-1}(N+j-i)\prod_{j=1}^{d-1}(L+j-i)}{\sum_{i=0}^{L}\frac{g_i^2-g_{i-1}^2}{d-1}\prod_{j=1}^{d-1}(N+j-i)\prod_{j=1}^{d-1}(L+j-i)}.
\end{equation}
\end{lemma}

\begin{lemma}\label{lemma-882357}
Using the instantiation of $g_i$ in \cref{eq-882338}, we can show that \cref{eq-880122} equals 
\begin{equation}\label{eq-10252149}
\frac{d-1}{L+N+d+\frac{2}{d+1}NL}.
\end{equation}
\end{lemma}

Then we can prove our \cref{thm-8100202}.
\begin{proof}[Proof of \cref{thm-8100202}]
By \cref{lemma-882357} and noting that $N=(d+1)L$, $L=\frac{n}{2d}$, we can see that the expectation of infidelity is
\begin{align}
\cref{eq-10252149}&=O\left(\frac{d}{L(L+d)}\right)=O\left(\frac{d^3}{n(n+d^2)}\right).\nonumber
\end{align}
\end{proof}

\section{Deferred lemmas and proofs}
\subsection{Some facts}
\begin{fact}\label{lemma-10201504}
Suppose $\ket{\psi}$, $\ket{\varphi}$ are two pure states on a $d$-dimensional Hilbert space $\mathcal{H}$, and $\Pi$ is an orthogonal projector onto a subspace of $\mathcal{H}$. Then we have
\[\left|\sqrt{\bra{\psi}\Pi\ket{\psi}}-\sqrt{\bra{\varphi}\Pi\ket{\varphi}}\right|\leq \frac{1}{\sqrt{2}}\|\ketbra{\psi}{\psi}-\ketbra{\varphi}{\varphi}\|_1,\]
where $\|\cdot \|_1$ is the trace norm.
\end{fact}
\begin{proof}
Let $r e^{i\theta}=\braket{\varphi}{\psi}$. 
Then, we have
\begin{align}
\left|\sqrt{\bra{\psi}\Pi\ket{\psi}}-\sqrt{\bra{\varphi}\Pi\ket{\varphi}}\right|&=\Big|\left\|\Pi\ket{\psi}\right\|-\left\|e^{i\theta}\Pi\ket{\varphi}\right\|\Big|\nonumber \\
&\leq \left\|\Pi\left(\ket{\psi}-e^{i\theta}\ket{\varphi}\right)\right\|\nonumber \\
&\leq \|\ket{\psi}-e^{i\theta}\ket{\varphi}\| \nonumber \\
&=\sqrt{2-2\Re(e^{-i\theta}\braket{\varphi}{\psi})} \nonumber \\
&=\sqrt{2-2|\braket{\varphi}{\psi}|} \nonumber \\
&\leq \sqrt{2-2|\braket{\varphi}{\psi}|^2} \nonumber \\
&=\frac{1}{\sqrt{2}}\|\ketbra{\psi}{\psi}-\ketbra{\varphi}{\varphi}\|_1,\nonumber
\end{align}
where $\|\ket{\psi}\|=\sqrt{\braket{\psi}{\psi}}$ is the vector $2$-norm.
\end{proof}

\begin{fact}\label{fact-872212}
For $0\leq i\leq L$, we have
\[\alpha_i\sqrt{\frac{\dim(\mathcal{Q}^d_{\gamma_i})}{\dim(\mathcal{Q}^d_{\gamma_i^+})}}=\sqrt{\frac{(N+d-i-1)(N-i+1)}{(L+N+d-2i-1)(L+N+d-2i)}},\]
and for $1\leq i\leq L$, we have
\[\beta_{i-1}\sqrt{\frac{\dim(\mathcal{Q}^d_{\gamma_{i-1}})}{\dim(\mathcal{Q}^d_{\gamma_{i}^+})}}=\sqrt{\frac{(L-i+1)(L+d-i-1)}{(L+N+d-2i+1)(L+N+d-2i)}}.\]
\end{fact}
\begin{proof}
By \cref{eq-10271256}, we have
\begin{align}
\frac{\dim(\mathcal{Q}^d_{\gamma_i})}{\dim(\mathcal{Q}^d_{\gamma_i^+})}&=\frac{(L+N+d-2i-1)\cdot \prod_{2\leq j\leq d-1}(N+j-i-1)}{(L+N+d-2i)\cdot \prod_{2\leq j\leq d-1}(N+j-i)}\nonumber\\
&=\frac{(L+N+d-2i-1)(N-i+1)}{(L+N+d-2i)(N+d-i-1)}.\nonumber
\end{align}
\begin{align}
\frac{\dim(\mathcal{Q}^d_{\gamma_{i-1}})}{\dim(\mathcal{Q}^d_{\gamma_{i}^+})}&=\frac{(L+N+d-2i+1)\prod_{2\leq j\leq d-1}(L+d-i-j+1)}{(L+N+d-2i)\prod_{2\leq j\leq d-1}(L+d-i-j)}\nonumber\\
&=\frac{(L+N+d-2i+1)(L+d-i-1)}{(L+N+d-2i)(L-i+1)}.\nonumber
\end{align}
Then, combined with \cref{eq-872206}, we can see
\[\alpha_i\sqrt{\frac{\dim(\mathcal{Q}^d_{\gamma_i})}{\dim(\mathcal{Q}^d_{\gamma_i^+})}}=\sqrt{\frac{(N+d-i-1)(N-i+1)}{(L+N+d-2i-1)(L+N+d-2i)}},\]
and
\[\beta_{i-1}\sqrt{\frac{\dim(\mathcal{Q}^d_{\gamma_{i-1}})}{\dim(\mathcal{Q}^d_{\gamma_{i}^+})}}=\sqrt{\frac{(L-i+1)(L+d-i-1)}{(L+N+d-2i+1)(L+N+d-2i)}}.\]
\end{proof}

\begin{fact}\label{fact-882356}
For any integer $k\geq 0$ and number $a,b$, we have
\[(a+b+k+1)\prod_{j=1}^k (a+j) (b+j) = \frac{1}{k+1}\left(\prod_{j=1}^{k+1}(a+j)(b+j)-\prod_{j=1}^{k+1} (a-1+j)(b-1+j)\right).\]
\end{fact}
\begin{proof}
\begin{align}
\prod_{j=1}^{k+1}(a+j)(b+j)-\prod_{j=1}^{k+1} (a-1+j)(b-1+j)&=\prod_{j=1}^{k+1}(a+j)(b+j)-\prod_{j=0}^k (a+j)(b+j)\nonumber\\
&=((a+k+1)(b+k+1)-ab)\prod_{j=1}^k (a+j)(b+j)\nonumber\\
&=(k+1)(a+b+k+1)\prod_{j=1}^k(a+j)(b+j)\nonumber
\end{align}
\end{proof}

\subsection{Proof of \texorpdfstring{\cref{lemma-882333}}{Lemma 4.2}}
\begin{proof}[Proof of \cref{lemma-882333}]
From \cref{eq-10270121} and \cref{eq-882258}, and by direct calculation,
\[f_ix_i=g_i\cdot (N-i+1) \sqrt{\frac{\prod_{j=2}^{d-1}(N+j-i) \prod_{j=2}^{d-1}(L+d-j-i)}{L+N+d-2i}},\]
\[f_{i-1}y_i=g_{i-1}\cdot  (L+d-i-1) \sqrt{\frac{\prod_{j=2}^{d-1}(N+j-i) \prod_{j=2}^{d-1}(L+d-j-i)}{L+N+d-2i}}.\]

Therefore, \cref{eq-831743} can be written as
\begin{equation}\label{eq-871520}
\frac{\sum_{i=0}^L\frac{((N-i+1)g_i+(L+d-i-1)g_{i-1})^2}{L+N+d-2i} \prod_{j=2}^{d-1}(N+j-i) \prod_{j=2}^{d-1}(L+d-j-i)}{\sum_{i=0}^L g_i^2(L+N+d-2i-1)\prod_{j=2}^{d-1}(N+j-i-1)\prod_{j=2}^{d-1}(L+d-j-i)}.
\end{equation}
Note that the denominator of \cref{eq-871520} equals
\begin{align}
&\sum_{i=0}^L g_i^2(L+N+d-2i-1)\prod_{j=1}^{d-2}(N+j-i)\prod_{j=1}^{d-2}(L+j-i)\nonumber\\
=& \sum_{i=0}^L \frac{g_i^2}{d-1}\left(\prod_{j=1}^{d-1}(N+j-i)(L+j-i)-\prod_{j=1}^{d-1}(N+j-i-1)(L+j-i-1)\right)\label{eq-890029}\\
=&\sum_{i=0}^{L}\frac{g_i^2-g_{i-1}^2}{d-1}\prod_{j=1}^{d-1}(N+j-i)(L+j-i) - \frac{g_L^2}{d-1}\prod_{j=1}^{d-1}(N-L+j-1)(j-1)\nonumber\\
=&\sum_{i=0}^{L}\frac{g_i^2-g_{i-1}^2}{d-1}\prod_{j=1}^{d-1}(N+j-i)(L+j-i),\label{eq-880020}
\end{align}
where \cref{eq-890029} is by using \cref{fact-882356}.
Also note that the numerator of \cref{eq-871520} equals
\begin{subequations}
\begin{align}
&\sum_{i=0}^L\frac{((N-i+1)g_i+(L+d-i-1)g_{i-1})^2}{L+N+d-2i} \prod_{j=1}^{d-2}(N+j-i+1) \prod_{j=1}^{d-2}(L+j-i)\nonumber\\
=&\sum_{i=0}^L\frac{\left((L+N+d-2i)g_i-(L+d-i-1)(g_{i}-g_{i-1})\right)^2}{L+N+d-2i} \prod_{j=1}^{d-2}(N+j-i+1) \prod_{j=1}^{d-2}(L+j-i)\nonumber\\
=&\sum_{i=0}^L\left((L+N+d-2i)g_i^2-2(L+d-i-1)g_i(g_i-g_{i-1})+\frac{(L+d-i-1)^2(g_{i}-g_{i-1})^2}{L+N+d-2i}\right) \nonumber\\
&\qquad\qquad\qquad\qquad\qquad\qquad\qquad\qquad\qquad\qquad\qquad\qquad\qquad\qquad\cdot\prod_{j=1}^{d-2}(N+j-i+1) \prod_{j=1}^{d-2}(L+j-i)\nonumber\\
=&\sum_{i=0}^L (L+N+d-2i)g_i^2\prod_{j=1}^{d-2}(N+j-i+1) \prod_{j=1}^{d-2}(L+j-i) \label{eq-880004}\\
&+\sum_{i=0}^L\left(-2(L+d-i-1)g_i(g_i-g_{i-1})+\frac{(L+d-i-1)^2(g_{i}-g_{i-1})^2}{L+N+d-2i}\right) \prod_{j=1}^{d-2}(N+j-i+1) \prod_{j=1}^{d-2}(L+j-i). \label{eq-880005}
\end{align}
\end{subequations}
The part in \cref{eq-880004} equals
\begin{align}
&\sum_{i=0}^L (L+N+d-2i)g_i^2\prod_{j=1}^{d-2}(N+j-i+1) \prod_{j=1}^{d-2}(L+j-i)\nonumber\\
=&\sum_{i=0}^L \frac{g_i^2}{d-1} \left(\prod_{j=1}^{d-1}(N+j-i+1)(L+j-i)-\prod_{j=1}^{d-1}(N+j-i)(L+j-i-1)\right)\label{eq-890031}\\
=&\sum_{i=0}^L \frac{g_i^2-g_{i-1}^2}{d-1} \prod_{j=1}^{d-1}(N+j-i+1)(L+j-i)-\frac{g_L^2}{d-1}\prod_{j=1}^{d-1}(N-L+j)(j-1)\nonumber\\
=&\sum_{i=0}^L \frac{g_i^2-g_{i-1}^2}{d-1} \prod_{j=1}^{d-1}(N+j-i+1)(L+j-i),\label{eq-880015}
\end{align}
where \cref{eq-890031} is by using \cref{fact-882356}.
Then, the denominator of \cref{eq-871520} minus the numerator of \cref{eq-871520} is therefore $\eqref{eq-880020}-\eqref{eq-880015}-\eqref{eq-880005}$, which equals
\begin{align}
&\left(\sum_{i=0}^L \frac{g_i^2-g_{i-1}^2}{d-1}[(N+1-i)-(N+d-i)]\prod_{j=2}^{d-1}(N+j-i)\prod_{j=1}^{d-1}(L+j-i)\right)-\eqref{eq-880005}\nonumber\\
=&\left(\sum_{i=0}^L (g_{i-1}^2-g_{i}^2)\prod_{j=2}^{d-1}(N+j-i)\prod_{j=1}^{d-1}(L+j-i)\right)-\eqref{eq-880005}\nonumber\\
=&\sum_{i=0}^L \left((g_{i-1}^2-g_i^2)(L+d-i-1)+2(L+d-i-1)g_i(g_i-g_{i-1})-\frac{(L+d-i-1)^2(g_i-g_{i-1})^2}{L+N+d-2i} \right)\nonumber\\
&\qquad\qquad\qquad\qquad\qquad\qquad\qquad\qquad\qquad\qquad\qquad\qquad\qquad\qquad\qquad\qquad\cdot\prod_{j=2}^{d-1}(N+j-i)\prod_{j=1}^{d-2}(L+j-i)\nonumber\\
=&\sum_{i=0}^L \frac{(L+d-i-1)(N-i+1)(g_i-g_{i-1})^2}{L+N+d-2i}\prod_{j=2}^{d-1}(N+j-i)\prod_{j=1}^{d-2}(L+j-i)\nonumber\\
=&\sum_{i=0}^L \frac{(g_i-g_{i-1})^2}{L+N+d-2i} \prod_{j=1}^{d-1}(N+j-i)\prod_{j=1}^{d-1}(L+j-i).\label{eq-880053}
\end{align}
Then, $1- \cref{eq-871520}$ is therefore $\eqref{eq-880053}\div \eqref{eq-880020}$, which equals
\begin{equation*}
\frac{\sum_{i=0}^L \frac{(g_i-g_{i-1})^2}{L+N+d-2i} \prod_{j=1}^{d-1}(N+j-i)\prod_{j=1}^{d-1}(L+j-i)}{\sum_{i=0}^{L}\frac{g_i^2-g_{i-1}^2}{d-1}\prod_{j=1}^{d-1}(N+j-i)\prod_{j=1}^{d-1}(L+j-i)}.
\end{equation*}
\end{proof}

\subsection{Proof of \texorpdfstring{\cref{lemma-882357}}{Lemma 4.3}}
\begin{proof}[Proof of \cref{lemma-882357}]
It is easy to see that for $0\leq i\leq L$
\[g_{i}-g_{i-1}=L+N+d-2i.\]
Thus, \cref{eq-880122} equals
\begin{align}
&\frac{(d-1)\sum_{i=0}^L (L+N+d-2i) \prod_{j=1}^{d-1}(N+j-i)\prod_{j=1}^{d-1}(L+j-i)}{\sum_{i=0}^{L}(g_i+g_{i-1})(L+N+d-2i)\prod_{j=1}^{d-1}(N+j-i)\prod_{j=1}^{d-1}(L+j-i)}.\label{eq-880150}
\end{align}
The numerator of \cref{eq-880150} is
\begin{align}
&(d-1)\sum_{i=0}^L (L+N+d-2i) \prod_{j=1}^{d-1}(N+j-i)\prod_{j=1}^{d-1}(L+j-i)\nonumber\\
=&(d-1)\sum_{i=0}^L \frac{1}{d}\left(\prod_{j=1}^d(N+j-i)(L+j-i)-\prod_{j=1}^d(N+j-i-1)(L+j-i-1)\right) \label{eq-890033}\\
=&\frac{(d-1)}{d}\left(\prod_{j=1}^d(N+j)(L+j)-\prod_{j=1}^d(N-L+j-1)(j-1)\right) \nonumber\\
=&\frac{(d-1)}{d}\prod_{j=1}^d(N+j)(L+j), \label{eq-10271336}
\end{align}
where \cref{eq-890033} is by using \cref{fact-882356}.
The denominator of \cref{eq-880150} is
\begin{align}
&\sum_{i=0}^{L}(g_i+g_{i-1})(L+N+d-2i)\prod_{j=1}^{d-1}(N+j-i)\prod_{j=1}^{d-1}(L+j-i) \nonumber\\
=&\sum_{i=0}^L\frac{g_i+g_{i-1}}{d}\left(\prod_{j=1}^d(N+j-i)(L+j-i)-\prod_{j=1}^d(N+j-i-1)(L+j-i-1)\right)\label{eq-890037}\\
=&\frac{g_0}{d}\prod_{j=1}^d(N+j)(L+j)-\frac{g_L+g_{L-1}}{d}\prod_{j=1}^d(N-L+j-1)(j-1)+\sum_{i=1}^L \frac{g_i-g_{i-2}}{d} \prod_{j=1}^d(N+j-i)(L+j-i) \nonumber\\
=&\frac{g_0}{d}\prod_{j=1}^d(N+j)(L+j)+\sum_{i=1}^L \frac{g_i-g_{i-2}}{d} \prod_{j=1}^d(N+j-i)(L+j-i) \nonumber\\
=&\frac{g_0}{d}\prod_{j=1}^d(N+j)(L+j)+\frac{2}{d}\sum_{i=1}^L (L+N+d-2i+1) \prod_{j=1}^d(N+j-i)(L+j-i)\nonumber\\
=&\frac{g_0}{d}\prod_{j=1}^d(N+j)(L+j)+\frac{2}{d(d+1)}\sum_{i=1}^L \left(\prod_{j=1}^{d+1}(N+j-i)(L+j-i)-\prod_{j=1}^{d+1}(N+j-i-1)(L+j-i-1)\right) \label{eq-890038}\\
=&\frac{g_0}{d}\prod_{j=1}^d(N+j)(L+j) +\frac{2}{d(d+1)}\left(\prod_{j=1}^{d+1}(N+j-1)(L+j-1)-\prod_{j=1}^{d+1}(N-L+j-1)(j-1)\right)\nonumber\\
=&\frac{g_0}{d}\prod_{j=1}^d(N+j)(L+j) +\frac{2}{d(d+1)}\prod_{j=1}^{d+1}(N+j-1)(L+j-1),\label{eq-10271338}
\end{align}
where \cref{eq-890037} and \cref{eq-890038} are by using \cref{fact-882356}.
Therefore, \cref{eq-880150} equals $\eqref{eq-10271336} \div\eqref{eq-10271338}$, which is
\[\frac{d-1}{g_0+\frac{2}{d+1}NL}=\frac{d-1}{L+N+d+\frac{2}{d+1}NL}.\]
\end{proof}

\section*{Acknowledgments}

We thank Ewin Tang and John Wright for valuable discussions.

\bibliographystyle{alpha}
\bibliography{main}

\newcommand{\etalchar}[1]{$^{#1}$}
\begin{thebibliography}{vACGN23}

\bibitem[BCH06]{bacon2006efficient}
Dave Bacon, Isaac~L Chuang, and Aram~W Harrow.
\newblock Efficient quantum circuits for schur and clebsch-gordan transforms.
\newblock {\em Physical review letters}, 97(17):170502, 2006.

\bibitem[BCH07]{bacon2007quantum}
Dave Bacon, Isaac~L Chuang, and Aram~W Harrow.
\newblock The quantum schur and clebsch-gordan transforms: I. efficient qudit circuits.
\newblock In {\em Proceedings of the eighteenth annual ACM-SIAM symposium on Discrete algorithms}, pages 1235--1244, 2007.

\bibitem[Bel75a]{belavkin1975optimal}
Viacheslav~P Belavkin.
\newblock Optimal multiple quantum statistical hypothesis testing.
\newblock {\em Stochastics: an international journal of probability and stochastic processes}, 1(1-4):315--345, 1975.

\bibitem[Bel75b]{belavkin1975optimum}
VP~Belavkin.
\newblock Optimum distinction of non-orthogonal quantum signals.
\newblock {\em Radiotekhnika i Elektronika}, 20:1177--1185, 1975.

\bibitem[BHMT02]{brassard2000quantum}
Gilles Brassard, Peter Hoyer, Michele Mosca, and Alain Tapp.
\newblock Quantum amplitude amplification and estimation.
\newblock {\em Contemporary Mathematics}, 2002.

\bibitem[BK02]{barnum2002reversing}
Howard Barnum and Emanuel Knill.
\newblock Reversing quantum dynamics with near-optimal quantum and classical fidelity.
\newblock {\em Journal of Mathematical Physics}, 43(5):2097--2106, 2002.

\bibitem[BM99]{bruss1999optimal}
Dagmar Bru{\ss} and Chiara Macchiavello.
\newblock Optimal state estimation for d-dimensional quantum systems.
\newblock {\em Physics Letters A}, 253(5-6):249--251, 1999.

\bibitem[CGJ19]{chakraborty2019ICALP}
Shantanav Chakraborty, Andr\'{a}s Gily\'{e}n, and Stacey Jeffery.
\newblock {The Power of Block-Encoded Matrix Powers: Improved Regression Techniques via Faster Hamiltonian Simulation}.
\newblock In {\em 46th International Colloquium on Automata, Languages, and Programming (ICALP 2019)}, volume 132, pages 33:1--33:14, 2019.

\bibitem[CML{\etalchar{+}}24]{chen2024quantum}
Yu-Ao Chen, Yin Mo, Yingjian Liu, Lei Zhang, and Xin Wang.
\newblock Quantum algorithm for reversing unknown unitary evolutions.
\newblock {\em arXiv preprint arXiv:2403.04704}, 2024.

\bibitem[CSM23]{cotler2023information}
Jordan Cotler, Thomas Schuster, and Masoud Mohseni.
\newblock Information-theoretic hardness of out-of-time-order correlators.
\newblock {\em Physical Review A}, 108(6):062608, 2023.

\bibitem[CYZ25]{chen2025tight}
Kean Chen, Nengkun Yu, and Zhicheng Zhang.
\newblock Tight bound for quantum unitary time-reversal.
\newblock {\em arXiv preprint arXiv:2507.05736}, 2025.

\bibitem[EGH{\etalchar{+}}11]{etingof2011introduction}
Pavel~I Etingof, Oleg Golberg, Sebastian Hensel, Tiankai Liu, Alex Schwendner, Dmitry Vaintrob, and Elena Yudovina.
\newblock {\em Introduction to representation theory}, volume~59.
\newblock American Mathematical Soc., 2011.

\bibitem[FH13]{fulton2013representation}
William Fulton and Joe Harris.
\newblock {\em Representation Theory: A First Course}, volume 129 of {\em Graduate Texts in Mathematics}.
\newblock Springer, 2013.

\bibitem[FK18]{fefferman2015quantum}
Bill Fefferman and Shelby Kimmel.
\newblock {Quantum vs. Classical Proofs and Subset Verification}.
\newblock In {\em 43rd International Symposium on Mathematical Foundations of Computer Science (MFCS 2018)}, volume 117, pages 22:1--22:23, 2018.

\bibitem[GBO24]{grinko2023gelfand}
Dmitry Grinko, Adam Burchardt, and Maris Ozols.
\newblock Gelfand-tsetlin basis for partially transposed permutations, with applications to quantum information.
\newblock In {\em 27th Conference on Quantum Information Processing}, 2024.

\bibitem[GKKT20]{guctua2020fast}
Madalin Gu{\c{t}}{\u{a}}, Jonas Kahn, Richard Kueng, and Joel~A Tropp.
\newblock Fast state tomography with optimal error bounds.
\newblock {\em Journal of Physics A: Mathematical and Theoretical}, 53(20):204001, 2020.

\bibitem[GL25]{grewal2025query}
Sabee Grewal and Daniel Liang.
\newblock Query-optimal estimation of unitary channels via pauli dimensionality.
\newblock {\em arXiv preprint arXiv:2510.00168}, 2025.

\bibitem[GSLW19]{gilyen2019quantum}
Andr{\'a}s Gily{\'e}n, Yuan Su, Guang~Hao Low, and Nathan Wiebe.
\newblock Quantum singular value transformation and beyond: exponential improvements for quantum matrix arithmetics.
\newblock In {\em Proceedings of the 51st annual ACM SIGACT symposium on theory of computing}, pages 193--204, 2019.

\bibitem[GW09]{goodman2009symmetry}
Roe Goodman and Nolan~R Wallach.
\newblock {\em Symmetry, representations, and invariants}, volume 255.
\newblock Springer, 2009.

\bibitem[Har05]{harrow2005applications}
Aram~W Harrow.
\newblock Applications of coherent classical communication and the schur transform to quantum information theory.
\newblock {\em arXiv preprint quant-ph/0512255}, 2005.

\bibitem[HHJ{\etalchar{+}}16]{haah2016sample}
Jeongwan Haah, Aram~W Harrow, Zhengfeng Ji, Xiaodi Wu, and Nengkun Yu.
\newblock Sample-optimal tomography of quantum states.
\newblock In {\em Proceedings of the forty-eighth annual ACM symposium on Theory of Computing}, pages 913--925, 2016.

\bibitem[HKOT23]{haah2023query}
Jeongwan Haah, Robin Kothari, Ryan O’Donnell, and Ewin Tang.
\newblock Query-optimal estimation of unitary channels in diamond distance.
\newblock In {\em 2023 IEEE 64th Annual Symposium on Foundations of Computer Science (FOCS)}, pages 363--390. IEEE, 2023.

\bibitem[HW94]{hausladen1994pretty}
Paul Hausladen and William~K Wootters.
\newblock A 'pretty good' measurement for distinguishing quantum states.
\newblock {\em Journal of Modern Optics}, 41(12):2385--2390, 1994.

\bibitem[Kah07]{kahn2007fast}
Jonas Kahn.
\newblock Fast rate estimation of a unitary operation in su (d).
\newblock {\em Physical Review A—Atomic, Molecular, and Optical Physics}, 75(2):022326, 2007.

\bibitem[Key06]{keyl2006quantum}
Michael Keyl.
\newblock Quantum state estimation and large deviations.
\newblock {\em Reviews in Mathematical Physics}, 18(01):19--60, 2006.

\bibitem[KW99]{keyl1999optimal}
Michael Keyl and Reinhard~F Werner.
\newblock Optimal cloning of pure states, testing single clones.
\newblock {\em Journal of Mathematical Physics}, 40(7):3283--3299, 1999.

\bibitem[LC19]{low2019hamiltonian}
Guang~Hao Low and Isaac~L Chuang.
\newblock Hamiltonian simulation by qubitization.
\newblock {\em Quantum}, 3:163, 2019.

\bibitem[Ngu24]{nguyen2023mixed}
Quynh~T Nguyen.
\newblock The mixed schur transform: efficient quantum circuit and applications.
\newblock In {\em 27th Conference on Quantum Information Processing}, 2024.

\bibitem[OW16]{o2016efficient}
Ryan O'Donnell and John Wright.
\newblock Efficient quantum tomography.
\newblock In {\em Proceedings of the forty-eighth annual ACM symposium on Theory of Computing}, pages 899--912, 2016.

\bibitem[OW17]{o2017efficient}
Ryan O'Donnell and John Wright.
\newblock Efficient quantum tomography ii.
\newblock In {\em Proceedings of the 49th Annual ACM SIGACT Symposium on Theory of Computing}, pages 962--974, 2017.

\bibitem[OYM24]{odake2024analytical}
Tatsuki Odake, Satoshi Yoshida, and Mio Murao.
\newblock Analytical lower bound on query complexity for transformations of unknown unitary operations.
\newblock {\em arXiv preprint arXiv:2405.07625}, 2024.

\bibitem[PSW25]{pelecanos2025debiased}
Angelos Pelecanos, Jack Spilecki, and John Wright.
\newblock The debiased keyl's algorithm: a new unbiased estimator for full state tomography.
\newblock {\em arXiv preprint arXiv:2510.07788}, 2025.

\bibitem[SHH25]{schuster2024random}
Thomas Schuster, Jonas Haferkamp, and Hsin-Yuan Huang.
\newblock Random unitaries in extremely low depth.
\newblock {\em Science}, 389(6755):92--96, 2025.

\bibitem[SSW25]{scharnhorst2025optimal}
Thilo Scharnhorst, Jack Spilecki, and John Wright.
\newblock Optimal lower bounds for quantum state tomography.
\newblock {\em arXiv preprint arXiv:2510.07699}, 2025.

\bibitem[TW25]{tang2025amplitude}
Ewin Tang and John Wright.
\newblock Amplitude amplification and estimation require inverses.
\newblock {\em arXiv preprint arXiv:2507.23787}, 2025.

\bibitem[vACGN23]{van2023quantum}
Joran van Apeldoorn, Arjan Cornelissen, Andr{\'a}s Gily{\'e}n, and Giacomo Nannicini.
\newblock Quantum tomography using state-preparation unitaries.
\newblock In {\em Proceedings of the 2023 annual ACM-SIAM symposium on discrete algorithms (SODA)}, pages 1265--1318. SIAM, 2023.

\bibitem[vAG19]{van2018improvements}
Joran van Apeldoorn and Andr\'{a}s Gily\'{e}n.
\newblock {Improvements in Quantum SDP-Solving with Applications}.
\newblock In {\em 46th International Colloquium on Automata, Languages, and Programming (ICALP 2019)}, volume 132, pages 99:1--99:15, 2019.

\bibitem[VK92]{Vilenkin1992}
N.~Ja. Vilenkin and A.~U. Klimyk.
\newblock Representations in the gel'fand-tsetlin basis and special functions.
\newblock In {\em Representation of Lie Groups and Special Functions: Volume 3: Classical and Quantum Groups and Special Functions}. Springer Netherlands, 1992.

\bibitem[Wer98]{werner1998optimal}
Reinhard~F Werner.
\newblock Optimal cloning of pure states.
\newblock {\em Physical Review A}, 58(3):1827, 1998.

\bibitem[YRC20]{yang2020optimal}
Yuxiang Yang, Renato Renner, and Giulio Chiribella.
\newblock Optimal universal programming of unitary gates.
\newblock {\em Physical review letters}, 125(21):210501, 2020.

\bibitem[Yue23]{yuen2023improved}
Henry Yuen.
\newblock An improved sample complexity lower bound for (fidelity) quantum state tomography.
\newblock {\em Quantum}, 7:890, 2023.

\end{thebibliography}

\end{document}